\documentclass[onecolumn,journal]{IEEEtran}
\ifCLASSINFOpdf
\else
\fi

\usepackage[T1]{fontenc}
\usepackage[latin9]{inputenc}
\usepackage{color}
\usepackage{array}
\usepackage{float}
\usepackage{mathrsfs}
\usepackage{mathtools}
\usepackage{amsthm}
\usepackage{amstext}
\usepackage{amssymb}
\usepackage{stmaryrd}
\usepackage{graphicx}
\usepackage{pgfplots}

\makeatletter

\theoremstyle{plain}
\newtheorem{thm}{\protect\theoremname}
\theoremstyle{plain}

\theoremstyle{plain}
\newtheorem{cor}[thm]{\protect\corollaryname}
\theoremstyle{plain}
\newtheorem{lem}[thm]{\protect\lemmaname}
\theoremstyle{definition}
\newtheorem{example}[thm]{\protect\examplename}
\theoremstyle{definition}

\theoremstyle{plain}
\newtheorem{remark}[thm]{\protect\remarkname}

\AtBeginDocument{
	
}

\makeatother

  \providecommand{\corollaryname}{Corollary}
  \providecommand{\examplename}{Example}
  \providecommand{\lemmaname}{Lemma}
  \providecommand{\propositionname}{Proposition}
  \providecommand{\theoremname}{Theorem}
  \providecommand{\definitionname}{Definition}
  \providecommand{\remarkname}{Remark}

\newcommand{\ceil}[1]{{\left\lceil {#1} \right\rceil}}
\newcommand{\floor}[1]{{\left\lfloor {#1} \right\rfloor}}

\newcommand{\Bgk}[1]{{\Bigg( {#1} \Bigg) }}
\begin{document}

%
\title{Weierstrass pure gaps from a Quotient
	of the Hermitian Curve	}
%
%
%

\author{Shudi~Yang and~Chuangqiang~Hu   
\thanks{S. Yang is with the School of Mathematical
	Sciences, Qufu Normal University, Shandong 273165, P.R.China. \protect\\
		\quad C. Hu is with the School of Mathematics, Sun Yat-sen University, Guangzhou 510275, P.R.China.\protect\\	
	\protect\\
	E-mail: huchq@mail2.sysu.edu.cn,~{yangshd3@mail2.sysu.edu.cn}} \protect\\
\thanks{Manuscript received *********; revised ********.}
}

\maketitle

\begin{abstract}
 In this paper, by employing the results over Kummer extensions, we give an arithmetic characterization of pure gaps at many totally ramified places over the quotients of Hermitian curves, including the well-studied Hermitian curves as special cases. The cardinality of these pure gaps is explicitly investigated. In particular, the numbers of gaps and pure gaps at a pair of distinct places are determined precisely, which can be regarded as an extension of the previous work by Matthews (2001) considered Hermitian curves. Additionally, some concrete examples are provided 
 to illustrate our results.

\end{abstract}


%
\IEEEpeerreviewmaketitle

\section{Introduction}
%
%
%
%

\IEEEPARstart{I}{n} \cite{goppa1977codes,goppa1981codes} Goppa constructed algebraic geometric codes (AG codes for short) from several rational places by using algebraic curves, which led to an important research line in coding theory. Nowadays a great deal of works are devoted to determining or improving the parameters of AG codes, see \cite{garcia1993consecutive,Maharaj2005riemann,maharaj2006floor,Homma2001Goppa,carvalho2005goppa} and the references therein. 

Weierstrass semigroups and pure gaps are of significant uses in the construction and analysis of AG codes 
 for their applications in obtaining codes with good parameters (see~\cite{HuYang2016Kummer,Masuda2}).  
In \cite{garcia1993consecutive,garcia1992goppa}, Garcia, Kim and Lax improved the Goppa bound using the arithmetical structure of the Weierstrass gaps at only one place.
Homma and Kim~\cite{Homma2001Goppa} introduced the concept of pure gaps and demonstrated a similar result 
 for a pair of places. 
And this was generalized to several places by Carvalho and Torres in~\cite{carvalho2005goppa}.

Weierstrass semigroups and gaps over specific Kummer extensions were well-studied in the literature. For instance, the authors of~\cite{KN,KN1,X} computed the Weierstrass gaps and improved the parameters of one-point AG codes from Hermitian curves.
Matthews~\cite{matthews2004weierstrass} investigated the Weierstrass semigroup of any collinear places on a Hermitian curve. In~\cite{matthews2005weierstrass}, Matthews generalized the results of~\cite{carvalho2005goppa,matthews2004weierstrass} by determining the Weierstrass semigroup of
arbitrary rational places on the quotient of the Hermitian curve defined by the equation $ y^m=x^q + x   $ over
$ \mathbb{F}_{q^2} $ where $ q $ is a prime power and $  m>2 $ is a divisor of $ q + 1 $.
For general Kummer extensions, the authors in \cite{Miriam2015,Masuda2} recently described the Weierstrass semigroups and gaps
at one place and two places. Bartoli, Quoos and Zini \cite{Bartoli2016} gave a criterion to find pure gaps at many places and presented families of pure gaps.  In \cite{Yang2017Wsemi,HuYang2016Kummer}, Hu and Yang explicitly determined the Weierstrass semigroups and pure gaps at many places, and constructed AG codes with excellent parameters. However, little is known about the numbers of gaps and pure gaps from algebraic curves. Kim \cite{kim} obtained lower bounds and upper bounds on the cardinality of gaps at two places, which was pushed forward by Homma \cite{homma1996weierstrass} to give the exact expression. Applying Homma's theory \cite{homma1996weierstrass}, the authors in \cite{matthews2001weierstrass,Bartoli2016} counted the gaps and pure gaps at two places on Hermitian curves and specific Kummer extensions, respectively.  

In this paper, we will study Weierstrass pure gaps of any $ n $-tuple rational totally ramified points on the quotients of the Hermitian curves, containing Hermitian curves as special cases. To be precise, we first focus our attention on characterizing the pure gaps over Hermitian curves and counting their total number. Then the number of pure gaps on the quotients of the Hermitian curves is determined. Besides, we explicitly present the numbers of gaps and pure gaps at two distinct rational points, extending the results on Hermitian curves maintained by Matthews \cite{matthews2001weierstrass}.

The remaider of the paper is organized as follows. In Section~\ref{sec:prelimi} we briefly recall some notations and preliminary results over arbitrary function fields. 
Section~\ref{sec:Hermitian} focuses on the pure gaps at many points over Hermitian curves, where the cardinality of pure gap set is determined explicitly through their arithmetic characterization.
Finally, in Section~\ref{sec:Quo}, we consider the quotient of the Hermitian curves and provide some examples by using our main results.

\section{Preliminaries}\label{sec:prelimi}

In this section, we introduce notations and present some basic facts on the  Weierstrass semigroups and gaps at distinct rational places over arbitrary function fields.

Let $ q  $ be a power of a prime $p$ and $ \mathbb{F}_{q} $ be a finite field with $ q $ elements.  We denote by $ F $ a function field with genus $ g $ over $ \mathbb{F}_q$ and by
$ \mathbb{P}_F $ the set of places of $ F $. The free abelian group generated by the places of $ F $ is denoted by $ \mathcal{D}_F $, whose element is called a divisor. Assume that  $ D=\sum_{P\in {\mathbb{P}_F}} n_P P$ is a divisor such that almost all $ n_P=0 $, then the degree of $ D $ is $ \deg(D)= \sum_{P\in \mathbb{P}_F} n_P $.
For a function $ f \in F $, the divisor of poles of $ f $ will be denoted by $ (f)_{\infty} $.

We introduce some notations concerning the Weierstrass semigroups. Given $ l $ distinct rational places of $ F $, named $ Q_1,\cdots, Q_l $, the Weierstrass semigroup
$ H(Q_1,\cdots, Q_l) $ is defined by
\[
\Big\{(s_1,\cdots, s_l)\in \mathbb{N}_0^l~\Big|~\exists f\in F~ \text{with}~ (f)_{\infty}=\sum_{i=1}^l s_i Q_i  \Big\},
\]
and the Weierstrass gap set $  G(Q_1,\cdots, Q_l)  $ is defined by $ \mathbb{N}_0^l \backslash H(Q_1,\cdots, Q_l) $, where $ \mathbb{N}_0 := \mathbb{N}\cup \{0\} $ denotes the set of nonnegative integers.

An important subset of the Weierstrass gap set is the set of pure gaps.
Homma and Kim~\cite{Homma2001Goppa} introduced the concept of pure gap set at a pair of rational places. Carvalho and Torres~\cite{carvalho2005goppa} generalized this notion to several rational places, denoted by $ G_0(Q_1,\cdots, Q_l) $, which is given by
\begin{align*}
\Big\{&(s_1,\cdots,s_l)\in \mathbb{N}^l~\Big|~\ell(G) = \ell(G -Q_j ) ~\text{for}~1\leqslant j \leqslant l, ~\text{where}~G=\sum_{i=1}^l s_iQ_i \Big\},
\end{align*}
where $ \ell(G) $ is the dimension of the Riemann-Roch space $ \mathcal{L}(G) $. In addition, they showed that $ (s_1,\cdots,s_l)  $ is a pure gap at $ (Q_1,\cdots, Q_l) $ if and only if
\begin{align*}
\ell(s_1Q_1+\cdots+s_l Q_l)=\ell((s_1-1)Q_1+\cdots+(s_l-1) Q_l).
\end{align*}

Let $ m\geqslant 2  $ be an integer coprime to $ q $. A Kummer extension is a field extension $ F/ \mathbb{F}_q(x)$ defined by an eqution $ y^m=f(x)^\lambda $, where $ f(x) =\prod_{i=1}^{r}(x-\alpha_i)\in \mathbb{F}_q(x) $, $ \gcd(m,r \lambda)=1 $ and the $ \alpha_i $'s are pairwise distinct elements in $ \mathbb{F}_q $.
 Recall that the rational function field $ F  $ has
 genus $ g= (r-1)(m-1)/2$. Let $ P_1,\cdots,P_r $ be the places associated to the zeros of $ x-\alpha_1,\cdots,x-\alpha_r $, respectively, and $ P_{\infty} $ be the unique place at infinity. It follows from~\cite{stichtenoth2009algebraic,niederreiter2001rational} that they are totally ramified in this extension. We will often make use of the following result
 by Hu and Yang \cite{HuYang2016Kummer}, where they built a characterization of the pure gaps at many places of a Kummer extension. 
\begin{thm}[\cite{HuYang2016Kummer}, Theorem 3]\label{thm:Weier}
	\label{lem:Weierstsemi2}
	Let $ P_1,\cdots,P_n $ and $ P_{\infty}  $ be totally ramified places over a Kummer extension defined above, where $ 2 \leqslant n \leqslant r $. The following assertions hold.
	\begin{enumerate}
		\item The pure gap set $ G_0(P_1,\cdots,P_n)$ is given by
		\begin{align*}
		&\Big\{(t_1,\cdots,t_n)\in \mathbb{N}^n~\Big |
		~ m \sum\limits_{s=1 \atop s\neq k}^n\left\lceil \dfrac{t_k-t_s}{m}\right\rceil \\
		&+ m (r-n) \left\lceil \dfrac{t_k}{m}\right\rceil
		> rt_k
		\textup{ for all } 1\leqslant k \leqslant n \Big\}.
		\end{align*}
		\item The pure gap set $	G_0(P_1,\cdots,P_{n-1},P_{\infty}) $ is given by
		\begin{align*}
		&\Big \{(t_1,\cdots,t_{n-1},t_n)\in \mathbb{N}^{n}~\Big |
		~ m \sum\limits_{i=2 }^{n-1}\left\lceil \dfrac{-at_n-t_i}{m}\right\rceil \\
		& 	+ m (r-n+1) \left\lceil \dfrac{-at_n}{m}\right\rceil ~ > t_1 + (a+bm)t_n, \\
		&~ m \sum\limits_{i=1 \atop i\neq j}^{n-1}\left\lceil \dfrac{t_j-t_i}{m}\right\rceil
		+ m (r-n+1) \left\lceil \dfrac{t_j}{m}\right\rceil
		> t_n+ rt_j
		\textup{ for all } 1\leqslant j \leqslant n-1
		\Big\} ,
		\end{align*}
		where $ a,b $ are integers satisfying $ ar+ bm=1 $.
	\end{enumerate}
\end{thm}

The following results are due to van Lint and Wilson \cite{vanLint2001combi} and Homma \cite{homma1996weierstrass}, which will be of use in the sequel.
\begin{lem}[\cite{vanLint2001combi}, Theorem 13.1]\label{lem:solu}
	The number of solutions of the equation
	\begin{equation*}
	x_1+\cdots+x_n=A
	\end{equation*}
	in nonnegative integers is 
	\begin{equation*}
	\left( \begin{array}{c} A+n-1 \\n-1 \end{array} \right).
	\end{equation*}
\end{lem}

\begin{thm}[\cite{homma1996weierstrass}, Theorem 1]\label{thm:Gap_2point}
	Let $ Q_1, Q_2 $ be any two distinct points on a smooth curve of genus $ g >1$. Then the number of gaps at $ Q_1, Q_2 $ is
	\begin{align*}
	\#G(Q_1,Q_2)=\sum_{\alpha_1\in G(Q_1) }\alpha_1 
	+ \sum_{\alpha_2\in G(Q_2) }\alpha_2 
	- \#G_0(Q_1,Q_2).
	\end{align*}
\end{thm}

\section{Pure gaps from a Hermitian Curve}\label{sec:Hermitian}
In this section, we consider a Hermitian curve defined by
$ y^{q+1}=x^q+x $ over $ \mathbb{F}_{q^2} $ and calculate the cardinality of the pure gaps at many places.

 The following theorem is an immediate consequence of Theorem \ref{lem:Weierstsemi2} by taking 
  $ m=q+1 $ and $ r=q $.
\begin{thm}
	For distinct rational places $ P_1 ,\cdots,  P_n $ on the Hermitian curve
	$ y^{q+1}=x^q+x $ over $ \mathbb{F}_{q^2} $, 
\begin{align}\label{eq:HermPurgap}
G_0(P_1,\cdots,P_n)=&\Big\{(t_1,\cdots,t_n)\in \mathbb{N}^n~\Big |
~ (q+1) \sum\limits_{s=1 \atop s\neq k}^n\left\lceil \dfrac{t_k-t_s}{q+1}\right\rceil \nonumber \\
&+ (q+1) (q-n) \left\lceil \dfrac{t_k}{q+1}\right\rceil
> qt_k
\textup{ for all } 1\leqslant k \leqslant n \Big\}.
\end{align}
\end{thm}

Before we count all the pure gaps at many places, we should give an arithmetic characterization of these pure gaps.
Consider one of the inequalities in~\eqref{eq:HermPurgap}, say
\begin{equation}\label{eq:oneconditon}
(q+1)\Bgk{\ceil{\dfrac{t_1-t_2}{q+1}} +\cdots+ \ceil{\dfrac{t_1-t_n}{q+1}}} +(q+1)(q-n)\ceil{\dfrac{t_1}{q+1}} >qt_1 .
\end{equation}
Dividing $ q+1 $ in both sides of \eqref{eq:oneconditon}
gives that
\begin{equation*} 
\ceil{\dfrac{t_1-t_2}{q+1}} +\cdots+ \ceil{\dfrac{t_1-t_n}{q+1}} +(q-n)\ceil{\dfrac{t_1}{q+1}} >\dfrac{q}{q+1}t_1 ,
\end{equation*} 
which means that
\begin{equation}\label{eq:oneconditon2}
\ceil{\dfrac{t_1-t_2}{q+1}} +\cdots+ \ceil{\dfrac{t_1-t_n}{q+1}} +(q-n)\ceil{\dfrac{t_1}{q+1}} \geqslant \ceil{\dfrac{q}{q+1}t_1} .
\end{equation} 
Putting $ t_k := (q+1)i_k + j_k $  where $ i_k \geqslant 0 $ and $ 1\leqslant j_k \leqslant q+1 $ for $ k=1,\cdots,n $, we obtain
\begin{equation}\label{eq:con1}
\ceil{\dfrac{j_1-j_2}{q+1}} +\cdots+ \ceil{\dfrac{j_1-j_n}{q+1}} +(q-n)\ceil{\dfrac{j_1}{q+1}} \geqslant A+\ceil{\dfrac{q}{q+1}j_1} 
\end{equation}
by~\eqref{eq:oneconditon2}. Here we write $ A:= \sum_{k=1}^{n} i_k $ for simplicity.
We claim that $ j_1\neq q+1 $. If not, we must have 
\begin{equation*}
(n-1)+(q-n)  \geqslant A+q
\end{equation*}
by~\eqref{eq:con1}. This gives a contradiction as $   A \geqslant 0 $ by definition. 
Thus $ 1 \leqslant j_1 \leqslant q $. In the same manner as above we have $ 1 \leqslant j_k \leqslant q $ for $ k=1,\cdots,n $. Since $ 1\leqslant j_1 \leqslant q $, it follows from~\eqref{eq:con1} that
\begin{equation*} 
\ceil{\dfrac{j_1-j_2}{q+1}} +\cdots+ \ceil{\dfrac{j_1-j_n}{q+1}} -j_1  \geqslant A-(q-n).
\end{equation*}
Moreover, it is easily verified that
\begin{equation*} 
\ceil{\dfrac{j_1-j_2}{q+1}} +\cdots+ \ceil{\dfrac{j_1-j_n}{q+1}} -j_1  \leqslant  (n-1)-2. 
\end{equation*}    Thus $ (n-1)-2 \geqslant A-(q-n) $, and so
 $  0 \leqslant A \leqslant q-3 $. 
Therefore, the pure gap set $ G_0(P_1,\cdots,P_n) $ can be represented as follows:
\begin{equation*}
G_0(P_1,\cdots,P_n) = \cup_{0\leqslant A \leqslant q-3} \Bgk{(q+1)(i_1,\cdots,i_n)+ J_A} ,
\end{equation*} 
where $ A= \sum_{k=1}^{n} i_k $ and the set $ J_A $ is defined by
\begin{align}\label{eq:Bjn}
J_A := 
\Big\{ &(j_1,\cdots,j_n)\in \mathbb{N}^n~\Big | 
~1 \leqslant j_1,\cdots,j_n \leqslant q , \nonumber \\
& \sum\limits_{s=1 \atop s\neq k}^n\ceil{\dfrac{j_k-j_s}{q+1}} -j_k  \geqslant A-(q-n) \textup{ for all } 1\leqslant k \leqslant n
\Big\}.
\end{align}

To calculate the number of lattice points in $ G_0(P_1,\cdots,P_n) $, we come to a characterization of $ J_A $. Without loss of generality, we suppose that $ j_1 \geqslant j_2 \geqslant \cdots \geqslant j_n $.
Note that the inequality 
\begin{align*}
\sum\limits_{s=1 \atop s\neq k}^n\ceil{\dfrac{j_k-j_s}{q+1}} -j_k  \geqslant A-(q-n)
\end{align*}
implies that
\begin{align}\label{eq:jk1}
1 \leqslant j_k \leqslant \sum\limits_{s=1 \atop s\neq k}^n\ceil{\dfrac{j_k-j_s}{q+1}} +(q-n)-A.
\end{align}
Under the assumption that $ j_1 \geqslant j_2 \geqslant \cdots \geqslant j_n $, we observe that   
\begin{align*}
\sum\limits_{s=1 \atop s\neq k}^n\ceil{\dfrac{j_k-j_s}{q+1}} \leqslant (k-1)\times 0+(n-k)\times 1 =n-k.
\end{align*} 
Hence from \eqref{eq:jk1}, we have for $ k=1,\cdots,n $,  
\begin{equation} \label{eq:jk}
1\leqslant j_k \leqslant q-k-A. 
\end{equation} 
It follows that $0 \leqslant A \leqslant q-n-1 $ by taking $ k=n $ in~\eqref{eq:jk}. By setting $ t:=q-A $, we denote
\begin{equation*}  
\mathcal{B}_t := \Big\{(j_1,\cdots,j_n)\Big|~j_1 \geqslant j_2 \geqslant \cdots \geqslant j_n 
\textup{ and } 1 \leqslant j_k \leqslant t-k \textup{ for all } 1\leqslant k \leqslant n \Big\}.
\end{equation*} 
Let $ \sigma $ be a permutation of the set $ \{1,\cdots,n\} $ and define 
\begin{equation*}
	\sigma (j_1,\cdots,j_n) := (j_{\sigma(1)},\cdots,j_{\sigma(n)}).
\end{equation*}
We denote by $ \tau_n $ the collection of all the permutations of the set $ \{1,\cdots,n\} $. Obviously $\# \tau_n=n! $. Since $ j_1,\cdots,j_n $ in the set $ J_A $ of \eqref{eq:Bjn} are symmetric, we can represent  $ J_A $ as the union set
$   \cup_{\sigma\in \tau_n} \sigma( \mathcal{B}_t)$, where 
\begin{equation*}
	\sigma( \mathcal{B}_t)
	:= \Big\{\sigma (j_1,\cdots,j_n) ~ \Big| ~ (j_1,\cdots,j_n)\in  \mathcal{B}_t \Big\}.
\end{equation*}
 Define
\begin{equation}\label{eq:Bt}
B_t:= \Big\{(j_1,\cdots,j_n)\Big|~  1 \leqslant j_k \leqslant t-k \textup{ for all } 1\leqslant k \leqslant n \Big\}.
\end{equation}
In particular, $ B_t=\varnothing $ if $ t\leqslant n $. One sees that the set $   \cup_{\sigma\in \tau_n} \sigma( \mathcal{B}_t)$ is equivalent to the set $   \cup_{\sigma\in \tau_n} \sigma(  B_t)$. In conclusion, we have 
\begin{align*}
J_A=\cup_{\sigma\in \tau_n} \sigma( \mathcal{B}_t)=\cup_{\sigma\in \tau_n} \sigma(B_t).
\end{align*}
Therefore, we can represent $ G_0(P_1,\cdots,P_n) $ as follows:
\begin{equation}\label{eq:G0}
G_0(P_1,\cdots,P_n) = \cup_{0\leqslant A \leqslant q-n-1} \Bgk{ (q+1)(i_1,\cdots,i_n)+  \cup_{\sigma\in \tau_n} \sigma(B_t)},
\end{equation} 
where $ t=q-A $ and $ A= \sum_{k=1}^{n} i_k $ with all $ i_k \geqslant 0 $.

Now we will compute the cardinality of $ G_0(P_1,\cdots,P_n) $. Let $ S_n(t) $ denote the number of lattice points in the union of $  \sigma(B_t) $ for all $ \sigma\in \tau_n $. Our key lemma is the following.
\begin{lem}\label{lem:Sn(t)}
	Set $ S_0(t)= 1 $. For $ 1 \leqslant n <q $, if $ t\geqslant n+1 $, then we have
	\begin{equation}\label{eq:Sn(t)}
	S_n(t) = (t-n )t^{n-1},
	\end{equation}
	and if $ t \leqslant n $,
	\begin{equation*}
	S_n(t) = 0.
	\end{equation*}
\end{lem}
\begin{proof}
   It is clear that $ B_t $ of~\eqref{eq:Bt} is empty provided that $ t\leqslant n $, which follows that $ S_n(t) = 0 $. 
   
   Suppose that $ t \geqslant n+1 $ in the reminder of the proof. To get \eqref{eq:Sn(t)}, we proceed by decreasing induction on $ n $.

	For $ n=1 $, $ S_1(t)=t-1 $.
	
	For $ n=2 $, 
	\begin{equation*} 
	B_t = \Big\{(j_1,j_2)\Big| ~ 1 \leqslant j_1 \leqslant t-1, 1 \leqslant j_2 \leqslant t-2 \Big\}.
	\end{equation*}
	To count the lattice points in the set $ \cup_{\sigma\in \tau_2} \sigma(B_t) $, we consider two sets
	\begin{equation*} 
	\varOmega_0= \Big\{(j_1,j_2)\Big| ~ 1 \leqslant j_k \leqslant t-2 \textup{ for } k=1,2 \Big\}
	\end{equation*}
	and
	\begin{equation*} 
	\varOmega_1= \Big\{(j_1,j_2)\Big| ~ j_1=t-1, ~1 \leqslant  j_2 \leqslant t-2  \Big\}.
	\end{equation*}	
	Then $ \cup_{\sigma\in \tau_2} \sigma(B_t) $ becomes $\varOmega_0 \cup (\cup_{\sigma\in \tau_2} \sigma(\varOmega_1))  $ since $ \cup_{\sigma\in \tau_2} \sigma(\varOmega_0)= \varOmega_0  $. Note that the intersection of $\varOmega_0 $  and $  \cup_{\sigma\in \tau_2} \sigma(\varOmega_1) $ is the empty set. Hence
	$ S_2(t)=(t-2)^2+2(t-2)=(t-2)t $, which settles the case where $  n = 2 $ (see Fig.~\ref{fig:Bt}).
	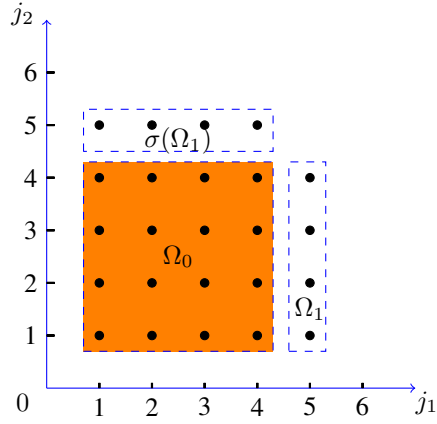
\begin{figure}[H]
		\centering
		\begin{tikzpicture}[scale=0.7]
		\path [fill=orange] (0.7,0.7)--(4.3,0.7)--(4.3,4.3)--(0.7,4.3)--(0.7,0.7);
		\draw [blue,->] (0,0)--(7,0);
		\draw [blue,->] (0,0)--(0, 7);
		\draw [blue,dashed] (0.7,0.7)--(4.3,0.7)--(4.3,4.3)--(0.7,4.3)--(0.7,0.7);
		\draw [blue,dashed] (4.6,0.7)--(5.3,0.7)--(5.3,4.3)--(4.6,4.3)--(4.6,0.7);
        \draw [blue,dashed] (0.7,4.5)--(0.7,5.3)--(4.3,5.3)--(4.3,4.5)--(0.7,4.5);
        
        \draw[color=black] node [yshift=-1.2ex,xshift=1ex] at (7,0) { $ j_1 $ };
        \draw[color=black] node [yshift=0.5ex,xshift=-2ex] at (0,7) { $ j_2 $ };
		\draw[color=black] node [yshift=-1.2ex,xshift=-2ex] at (0,0) { $ 0 $ };
		
    	\draw [thick] (1,0)--(1,.1)  node[yshift=-2ex,xshift=0ex]{1} -- (1,0);
    	\draw [thick] (2,0)--(2,.1) node[yshift=-2ex,xshift=0ex]{2} -- (2,0);
	    \draw [thick] (3,0)--(3,.1) node[yshift=-2ex,xshift=0ex]{3} -- (3,0);
	    \draw [thick] (4,0)--(4,.1) node[yshift=-2ex,xshift=0ex]{4} -- (4,0);
	    \draw [thick] (5,0)--(5,.1) node[yshift=-2ex,xshift=0ex]{5} -- (5,0);
	    \draw [thick] (6,0)--(6,.1) node[yshift=-2ex,xshift=0ex]{6} -- (6,0);

    	\draw [thick] (0,1)--(0.15,1)   node[yshift=0ex,xshift=-2ex]{1} -- (0,1);
	    \draw [thick] (0,2)--(0.15,2) node[yshift=0ex,xshift=-2ex]{2} -- (0,2);
	    \draw [thick] (0,3)--(0.15,3) node[yshift=0ex,xshift=-2ex]{3} -- (0,3);
	    \draw [thick] (0,4)--(0.15,4) node[yshift=0ex,xshift=-2ex]{4} -- (0,4);
	    \draw [thick] (0,5)--(0.15,5) node[yshift=0ex,xshift=-2ex]{5} -- (0,5);
	    \draw [thick] (0,6)--(0.15,6) node[yshift=0ex,xshift=-2ex]{6} -- (0,6);

		\draw[color=black] node at (2.5,2.5) { $ \Omega_0 $ };
		\draw[color=black] node at (5,1.5)    { $ \Omega_1 $ };
		\draw[color=black] node at (2.5,4.7) { $ \sigma(\Omega_1) $ };

		\draw [fill] (1, 1) circle [radius=0.08];
		\draw [fill] (1, 2) circle [radius=0.08];
		\draw [fill] (1, 3) circle [radius=0.08];
		\draw [fill] (1, 4) circle [radius=0.08];
		\draw [fill] (1, 5) circle [radius=0.08];
		\draw [fill] (2, 1) circle [radius=0.08];
		\draw [fill] (2, 2) circle [radius=0.08];
		\draw [fill] (2, 3) circle [radius=0.08];
		\draw [fill] (2, 4) circle [radius=0.08];
		\draw [fill] (2, 5) circle [radius=0.08];
		\draw [fill] (3, 1) circle [radius=0.08];
		\draw [fill] (3, 2) circle [radius=0.08];
		\draw [fill] (3, 3) circle [radius=0.08];
		\draw [fill] (3, 4) circle [radius=0.08];
		\draw [fill] (3, 5) circle [radius=0.08];
		\draw [fill] (4, 1) circle [radius=0.08];
		\draw [fill] (4, 2) circle [radius=0.08];
		\draw [fill] (4, 3) circle [radius=0.08];
		\draw [fill] (4, 4) circle [radius=0.08];
		\draw [fill] (4, 5) circle [radius=0.08];
		\draw [fill] (5, 1) circle [radius=0.08];
		\draw [fill] (5, 2) circle [radius=0.08];
		\draw [fill] (5, 3) circle [radius=0.08];
		\draw [fill] (5, 4) circle [radius=0.08]; 	
		\end{tikzpicture}
		
		\protect\caption{The lattice point set  $ \cup_{\sigma\in \tau_2} \sigma(B_t) $ for $ t=6 $ }
		\label{fig:Bt}
	\end{figure}

    We now proceed by induction on $ n \geqslant 3 $.
	Assume that $ 	S_l(t) = (t-l )t^{l-1} $ holds for all  $ 2 \leqslant l \leqslant n-1 $. For $ l=n $,
	\begin{equation*}
	B_t= \Big\{(j_1,\cdots,j_n)\Big|~  1 \leqslant j_k \leqslant t-k \textup{ for all } 1\leqslant k \leqslant n \Big\}.
	\end{equation*}
	Then the set $ \cup_{\sigma\in \tau_n} \sigma(B_t) $ contains all the $ n $-tuples $ (j_1,\cdots,j_n)$ with $  1 \leqslant j_k \leqslant t-k  $ for $ 1\leqslant k \leqslant n $ and all of their permutations $ \sigma(j_1,\cdots,j_n) $, $ \sigma\in \tau_n $. To count the lattice points in the set $ \cup_{\sigma\in \tau_n} \sigma(B_t) $, we consider the sets
	\begin{align*} 
	\varOmega_i= \Big\{(j_1,\cdots,j_n)\Big|& ~ t-n+1 \leqslant j_k \leqslant t-k  \textup{ for all } 1\leqslant k \leqslant i, \\
	& ~  1 \leqslant j_s \leqslant t-n  \textup{ for all } i+1\leqslant s \leqslant n
	\Big\},
	\end{align*}
	for $ i=0,1,\cdots,n-1 $.
	Observe that the sets $ \cup_{\sigma \in \tau_n} \sigma(\varOmega_i) $, with $ i=0,1,\cdots,n-1 $, are pairwise disjoint and  $ \cup_{\sigma \in \tau_n} \sigma(\varOmega_0) = \varOmega_0$. Thus we can partition 
	$ \cup_{\sigma\in \tau_n} \sigma(B_t) $ into $ n $ subsets  $ \cup_{\sigma \in \tau_n} \sigma(\varOmega_i) $ for $ i=0,1,\cdots,n-1 $.

 Let $ \rho_i $ denote the number of lattice points in $ \cup_{\sigma \in \tau_n} \sigma(\varOmega_i) $. Then $ \rho_0=(t-n)^n $. For $ i\geqslant 1 $, we claim that
	\begin{equation}\label{eq:rho_i}
	\rho_i=\left( \begin{array}{c} n \\i  \end{array} \right) S_i(n)(t-n)^{n-i}. 
	\end{equation}
	To see this, we separate each $ n $-tuple $ (j_1,\cdots,j_n) $ in $ \varOmega_i $ into two distinct parts $ (j_1,\cdots,j_i) $ and $ (j_{i+1},\cdots,j_n) $ such that $ t-n+1 \leqslant j_k \leqslant t-k  $ for all $ 1\leqslant k \leqslant i  $ and $  1 \leqslant j_s \leqslant t-n $ for all $ i+1\leqslant s \leqslant n $. 
	Thus we must choose $ i $ out of $ n $ positions. Define
	\begin{align*} 
	\varOmega_{i}^{(1)} &:= \Big\{(j_1,\cdots,j_i)\Big| ~ t-n+1 \leqslant j_k \leqslant t-k  \textup{ for all } 1\leqslant k \leqslant i 
	\Big\},\\
	\intertext{and }  
	\varOmega_{i}^{(0)} &:= \Big\{(j_{i+1},\cdots,j_n)\Big|  
	~  1 \leqslant j_s \leqslant t-n  \textup{ for all } i+1\leqslant s \leqslant n
	\Big\},
	\end{align*}
	for each $ \Omega_i $. Clearly $ \cup_{\sigma_0 \in \tau_{n-i}} \sigma_0(\varOmega_{i}^{(0)})= \varOmega_{i}^{(0)}$ and $ \varOmega_{i}^{(0)} $ has cardinality $(t-n)^{n-i} $. Meanwhile the set $ \varOmega_{i}^{(1)} $ is equivalent to
	 \begin{align*}\Big\{(j_1,\cdots,j_i)\Big| ~ 1 \leqslant j_k \leqslant n-k  \textup{ for all } 1\leqslant k \leqslant i 
	 \Big\},
	 \end{align*}
so $ \cup_{\sigma_1 \in \tau_i} \sigma_1(\varOmega_{i}^{(1)}) $ has cardinality $  S_i(n) $. The choice of $(j_1,\cdots,j_i) $ now leads to the equality
		\begin{align*} 
		\rho_i&=\left( \begin{array}{c} n \\i  \end{array} \right) 
		\cdot \#\Bgk{\cup_{\sigma_1 \in \tau_i} \sigma_1(\varOmega_{i}^{(1)})}
		\cdot \# \Bgk{\cup_{\sigma_0 \in \tau_{n-i}} \sigma_0(\varOmega_{i}^{(0)})}\\
		&=\left( \begin{array}{c} n \\i  \end{array} \right) 
		S_i(n)(t-n)^{n-i}, 
		\end{align*}
	concluding the claim of \eqref{eq:rho_i}.
	
	Using \eqref{eq:rho_i} and the induction hypothesis, we obtain
	\begin{align*}
	S_n(t)& = \sum_{i=0}^{n-1} \rho_i= \sum_{i=0}^{n-1} \left( \begin{array}{c} n \\i  \end{array} \right) S_i(n)(t-n)^{n-i}\\
   	& = (t-n)\sum_{i=0}^{n-1} \left( \begin{array}{c} n-1 \\i  \end{array} \right) n^i (t-n)^{n-1-i}\\
	& = (t-n)t^{n-1},
	\end{align*}
	which proves the lemma.		
\end{proof}

The cardinality of $ G_0(P_1,\cdots,P_n) $ on a Hermitian curve is given as follows.
\begin{thm}\label{th:Hermitian pure gap}
 For the pure gap set $ G_0(P_1,\cdots,P_n)$ on the Hermitian curve
 $ y^{q+1}=x^q+x $ over $ \mathbb{F}_{q^2} $, if $ n<q $, then
	\begin{equation*}
	\#  G_0(P_1,\cdots,P_n)= \sum_{A=0}^{q-n-1} \left( \begin{array}{c} A+n-1 \\n-1 \end{array} \right) (q-A-n)(q-A)^{n-1},
	\end{equation*}
	and if $ n\geqslant q $, $ G_0(P_1,\cdots,P_n)=\varnothing$.
\end{thm}
\begin{proof}
	It is trivial that $ G_0(P_1,\cdots,P_n)=\varnothing$ if $ n\geqslant q $. So we suppose that $ n<q $. From~\eqref{eq:G0} and Lemma~\ref{lem:solu}, the number of $ n  $-tuples $ (i_1,\cdots,i_n) $ such that $ i_1+\cdots+i_n=A $ is
	\begin{equation*}
	\left( \begin{array}{c} A+n-1 \\n-1 \end{array} \right).
	\end{equation*}
	In \eqref{eq:G0}, we substitute $ t$ by $q-A $. Then the desired conclusion follows immediately from Lemma~\ref{lem:Sn(t)} and \eqref{eq:G0}.
\end{proof}
\begin{remark}
	It is not difficult to compute the formula
	\begin{equation*}
	\# G_0(P_1,P_2) = \sum_{A=0}^{q-3}\Bgk{(q-1-A)^2-1}(A+1)=\dfrac{q}{12}(q-1)(q-2)(q+3),
	\end{equation*} 
	which appears in \cite{matthews2001weierstrass}. Thus one can regard Theorem \ref{th:Hermitian pure gap} as a generalization of this result concerning about arbitrary places.
\end{remark}

\section{Pure gaps from a Quotient of the Hermitian Curve }\label{sec:Quo}
In this section, we consider a quotient of the Hermitian curve defined by 
\begin{align}\label{eq:Quocurve}
 y^m=x^q+x
\end{align}  over $ \mathbb{F}_{q^2} $ where
$ m \geqslant 2 $ is a divisor of $ q +1 $. This curve was originally studied by Schmidt \cite{Schmidt1939} as the first known example of a non-classical curve.
Remember that the case $ m=q+1 $ is associated with the much-studied Hermitian curve. 

In the rest of this paper, we always assume that $ q+1=mN $ for a positive integer $ N $.  By Theorem~\ref{thm:Weier}, we obtain the pure gaps from
a quotient of the Hermitian curve.
\begin{thm}\label{thm:Quopuregap}
	Let $ P_1 ,\cdots,  P_n $ be distinct rational places and $ P_{\infty} $ be the unique place at infinity on the quotient of the Hermitian curve defined by \eqref{eq:Quocurve}. Then the pure gap sets $ G_0(P_1,\cdots,P_{n}) $ and $ G_0(P_1,\cdots,P_{n-1},P_{\infty}) $ 
	are both equal to each other, which is given by 
     \begin{align}\label{eq:HermPurgapQuo}
      &\Big\{(t_1,\cdots,t_n)\in \mathbb{N}^n~\Big |
       ~ m \sum\limits_{s=1 \atop s\neq k}^n\left\lceil \dfrac{t_k-t_s}{m}\right\rceil \nonumber \\
      &+ m (q-n) \left\lceil \dfrac{t_k}{m}\right\rceil
       > qt_k
       \textup{ for all } 1\leqslant k \leqslant n \Big\}.
   \end{align} 
\end{thm} 
\begin{proof}
   If we replace $ r $ by $ q $ in the first assertion of Theorem~\ref{thm:Weier}, then we get the pure gap set $ G_0(P_1,\cdots,P_{n}) $ as described in \eqref{eq:HermPurgapQuo}.

   Suppose that $ q+1=mN $ for an integer $ N$. So $ -q+Nm=1 $. Putting $ r=q $, $ a=-1 $ and $ b=N $ in the second assertion of Theorem~\ref{thm:Weier}, we obtain the pure gap set $	G_0(P_1,\cdots,P_{n-1},P_{\infty}) $ as follows:
  \begin{align}\label{eq:purePinf}
    &\Big \{(t_1,\cdots,t_{n-1},t_n)\in \mathbb{N}^{l}~\Big |
     ~ m \sum\limits_{s=2 }^{n-1}\left\lceil \dfrac{t_n-t_s}{m}\right\rceil \nonumber\\
    & 	+ m (q-n+1) \left\lceil \dfrac{ t_n}{m}\right\rceil ~ > t_1 + qt_n, \nonumber\\
    &~ m \sum\limits_{s=1 \atop s\neq k}^{n-1}\left\lceil \dfrac{t_k-t_s}{m}\right\rceil
    + m (q-n+1) \left\lceil \dfrac{t_k}{m}\right\rceil
    > t_n+ qt_k
   \textup{ for all } 1\leqslant k \leqslant n-1
   \Big\} .
  \end{align}	   
   It remains to prove that $	G_0(P_1,\cdots,P_{n}) $ equals $	G_0(P_1,\cdots,P_{n-1},P_{\infty}) $. Taking $ k=1 $ in  \eqref{eq:HermPurgapQuo} leads to
   \begin{equation}\label{eq:oneconditonQuo}
   \sum\limits_{s=2 }^{n} \ceil{\dfrac{t_1-t_s}{m}} +(q-n)\ceil{\dfrac{t_1}{m}} >\frac{qt_1}{m} .
   \end{equation}
   Write each $ t_k $ in terms of $  m i_k + j_k $, where $ i_k \geqslant 0 $ and $ 1\leqslant j_k \leqslant m $ for $ k=1,\cdots,n $. By \eqref{eq:oneconditonQuo}, we have
   \begin{equation}\label{eq:con1Quo}
   \ceil{\dfrac{j_1-j_2}{m}} +\cdots+ \ceil{\dfrac{j_1-j_n}{m}} +(q-n)\ceil{\dfrac{j_1}{m}} > A+ \dfrac{q}{m}j_1 ,
   \end{equation}  
   where $ A:= \sum_{k=1}^{n} i_k $.
   We claim that $ j_1\neq m $. Suppose that this is false, so $ j_1=m $ and
   \begin{equation*}
   (n-1)+(q-n)  > A+q
   \end{equation*}
   by~\eqref{eq:con1Quo}. Thus $ A<-1 $ giving a contradiction as $   A \geqslant 0 $ by definition.
   So we must have $ 1 \leqslant j_1 \leqslant m-1 $ and $ m \nmid t_1 $. (In the same manner as above, we must have $ 1 \leqslant j_k \leqslant m-1 $ for $ k=1,\cdots,n $).   
   It follows from \eqref{eq:oneconditonQuo} that
      \begin{equation*} 
      \sum\limits_{s=2 }^{n} \ceil{\dfrac{t_1-t_s}{m}} +(q-n)\ceil{\dfrac{t_1}{m}} > \floor{\frac{qt_1}{m}} =Nt_1+\floor{-\frac{t_1}{m}},
      \end{equation*}
      since $ q=mN-1 $. It is clear that $ \ceil{\dfrac{ t_n}{m}}+\floor{-\dfrac{t_n}{m}} =0 $. This establishes that
        \begin{align*}  
        \sum\limits_{s=2 }^{n-1} \ceil{ \dfrac{t_1-t_s}{m}} 
        +  (q-n+1)  \ceil{ \dfrac{ t_1}{m} } >
         \floor{ \dfrac{t_n-t_1}{m}} + Nt_1= \floor{\frac{t_n}{m} + \frac{qt_1}{m}},
        \end{align*}
   which means that
  \begin{align*}
  	  \sum\limits_{s=2 }^{n-1}\ceil{ \dfrac{t_1-t_s}{m}}
  	  	+  (q-n+1) \ceil{ \dfrac{ t_1}{m}} > \frac{t_n}{m} + \frac{qt_1}{m},
  \end{align*}  
   as appeared in \eqref{eq:purePinf}.
   This completes the proof of this theorem.
\end{proof}

In the following, we only count all the pure gaps in $	G_0(P_1,\cdots,P_{n}) $ on a quotient of the Hermitian curve, since the set 
$ G_0(P_1,\cdots,P_{n-1},P_{\infty}) $ is equal to $ G_0(P_1,\cdots,P_{n}) $.

As before, we denote $ t_k := m i_k + j_k $  where $ i_k \geqslant 0 $ and $ 1\leqslant j_k \leqslant m-1 $ for $ k=1,\cdots,n $. Since $ 1\leqslant j_1 \leqslant m-1 $ and $ \gcd(m,q)=1 $, it follows from~\eqref{eq:con1Quo} that
\begin{equation*}
\ceil{\dfrac{j_1-j_2}{m}} +\cdots+ \ceil{\dfrac{j_1-j_n}{m}} +(q-n)\ceil{\dfrac{j_1}{m}} \geqslant A+ \ceil{\dfrac{q}{m}j_1} ,
\end{equation*} 
which implies that
\begin{equation*} 
\ceil{\dfrac{j_1-j_2}{m}} +\cdots+ \ceil{\dfrac{j_1-j_n}{m}} - N j_1  \geqslant A-(q-n),
\end{equation*}
as $ q+1=mN $.
 Thus we get a set of lattice points $ J_A $ defined by 
 \begin{align*}
 J_A := 
 \Big\{ &(j_1,\cdots,j_n)\in \mathbb{N}^n~\Big | 
 ~1 \leqslant j_1,\cdots,j_n \leqslant m-1 , \nonumber \\
 & \sum\limits_{s=1 \atop s\neq k}^n\ceil{\dfrac{j_k-j_s}{m}} -N j_k  \geqslant A-(q-n) \textup{ for all } 1\leqslant k \leqslant n
 \Big\}.
 \end{align*}
 Now we come to a characterization of $ J_A $.
 Without loss of generality, we suppose that $ j_1 \geqslant j_2 \geqslant \cdots \geqslant j_n $.
 Note that the inequality 
 \begin{align*}
 \sum\limits_{s=1 \atop s\neq k}^n\ceil{\dfrac{j_k-j_s}{m}} -Nj_k  \geqslant A-(q-n)
 \end{align*}
 implies that
 \begin{align}\label{eq:Njk1}
 N \leqslant N j_k \leqslant \sum\limits_{s=1 \atop s\neq k}^n\ceil{\dfrac{j_k-j_s}{m}} +(q-n)-A.
 \end{align}
 Under the assumption that $ j_1 \geqslant j_2 \geqslant \cdots \geqslant j_n $, we observe that 
 \begin{align*}
 \sum\limits_{s=1 \atop s\neq k}^n\ceil{\dfrac{j_k-j_s}{m}}\leqslant (k-1)\times 0+(n-k)\times 1 =n-k.
 \end{align*}
  Hence from \eqref{eq:Njk1} we have for $ k=1,\cdots,n $,  
 \begin{equation} \label{eq:Njk}
 N\leqslant N j_k \leqslant q-k-A.
 \end{equation}
 It follows that $0 \leqslant A \leqslant q-n-N $ by taking $ k=n $ in~\eqref{eq:Njk}.
  Write $  q-A=N(t-1)+\beta $ where $ t:= \ceil{\dfrac{q-A}{N}} $ and $ 1 \leqslant \beta \leqslant N $. From~\eqref{eq:Njk} we find
  \begin{equation*} 
  1\leqslant  j_k \leqslant \floor{\dfrac{q-A-k}{N}}=t-1+ \floor{\dfrac{\beta-k}{N}}=t-\Bgk{\ceil{\dfrac{k-\beta}{N}}+1}.
  \end{equation*} 
  By setting $ \lambda_k:= \ceil{\dfrac{k-\beta}{N}}+1 $, we denote
  \begin{equation*}  
  \mathcal{B}_t: = \Big\{(j_1,\cdots,j_n)\Big|~j_1 \geqslant j_2 \geqslant \cdots \geqslant j_n 
  \textup{ and } 1 \leqslant j_k \leqslant t-\lambda_k \textup{ for all } 1\leqslant k \leqslant n \Big\}.
  \end{equation*}

  As before we can represent  $ J_A $ as the union set
  $   \cup_{\sigma\in \tau_n} \sigma( \mathcal{B}_t)$, where 
   \begin{equation*}
   \sigma( \mathcal{B}_t)
    := \Big\{\sigma (j_1,\cdots,j_n) ~ \Big| ~ (j_1,\cdots,j_n)\in  \mathcal{B}_t \Big\}.
    \end{equation*} 
    If we define
  \begin{equation}\label{eq:Bt1}
  B_t:= \Big\{(j_1,\cdots,j_n)\Big|~  1 \leqslant j_k \leqslant t-\lambda_k \textup{ for all } 1\leqslant k \leqslant n \Big\},
  \end{equation}
  then it is easily checked that the set $   \cup_{\sigma\in \tau_n} \sigma( \mathcal{B}_t)$ is equal to the set $   \cup_{\sigma\in \tau_n} \sigma(  B_t)$. Therefore, we have the following representation of $ G_0(P_1,\cdots,P_n) $:
  \begin{equation}\label{eq:G01}
  G_0(P_1,\cdots,P_n) = \cup_{0\leqslant A \leqslant q-n-N} \Bgk{ m(i_1,\cdots,i_n)+  \cup_{\sigma\in \tau_n} \sigma(B_t)},
  \end{equation} 
  where $ t= \ceil{\dfrac{q-A}{N}} $ and $ A= \sum_{k=1}^{n} i_k $ with all $ i_k \geqslant 0 $.

 Let $ S_n^A(t) $ denote the number of lattice points in the union of $ \sigma(B_t) $ for all $ \sigma\in \tau_n $. We have the following lemma.
  \begin{lem}\label{lem:Sn(t)1}
  	Let $ A,n,N,\beta $ be nonnegative integers such that $ 0\leqslant A \leqslant q-n-N $ and 
  	 $  q-A=N(t-1)+\beta $ where $ t= \ceil{\dfrac{q-A}{N}} $ and $ 1 \leqslant \beta \leqslant N $. Set $ S_0^A(t)= 1 $. If $ 1 \leqslant n \leqslant \beta $, we have
  	\begin{align*} 
  	S_n^A(t) = (t-1 )^n .
  	\end{align*}
  If $ n > \beta $, we have 
  	\begin{align*} 
  	S_n^A(t) = \left\{\begin{array}{ll}
  	0  & \textup{if } t\leqslant \lambda_n,\\
  	\sum_{i=0}^{n-1} 	
  	\left( \begin{array}{c} n \\i \end{array} \right) S_i^A(\lambda_n)(t-\lambda_n)^{n-i} & \textup{if } t > \lambda_n,
  	\end{array}  
  	\right.
  	\end{align*}
  	where $ \lambda_n= \ceil{\dfrac{n-\beta}{N}}+1 $.
  \end{lem}
  \begin{proof}
  	Suppose that $ n-\beta=N(\lambda_n-1)+b $ where $ \lambda_n= \ceil{\dfrac{n-\beta}{N}}+1 $ and $ 1 \leqslant b \leqslant N $. 
  	The set $ B_t $ of~\eqref{eq:Bt1} can be expressed as
  	 \begin{align}\label{eq:Bt2}
  	  B_t= \Big\{(j_1,\cdots,j_n) \Big|
  	  &~  1 \leqslant j_k \leqslant t-1 \textup{ for } 1\leqslant k \leqslant \beta,\nonumber \\
  	  &~  1 \leqslant j_k \leqslant t-2 \textup{ for } \beta+ 1\leqslant k \leqslant \beta +N,\nonumber \\
  	  &\cdots \nonumber \\
  	  &~  1 \leqslant j_k \leqslant t-\lambda_n \textup{ for } \beta+N(\lambda_n-1)+ 1\leqslant k \leqslant n
  	   \Big\}.
  	   \end{align}
  	   
  	 (i) If $ 1 \leqslant n \leqslant \beta $, we have
  	  \begin{equation*}
  	  B_t= \Big\{(j_1,\cdots,j_n)\Big|~  1 \leqslant j_k \leqslant t-1 \textup{ for all } 1\leqslant k \leqslant n \Big\}.
  	  \end{equation*}
  	  The desired assertion then follows.
  	  
  	(ii) If $ n > \beta $, it is trivial that $ B_t $ of~\eqref{eq:Bt2} is  empty if $ t\leqslant \lambda_n $. It remains to 
  	  to count the lattice points in the set $ \cup_{\sigma\in \tau_n} \sigma(B_t) $ for $ t>\lambda_n $. Let us consider the sets
  	\begin{align*} 
  	\varOmega_i= \Big\{(j_1,\cdots,j_n)\Big|& ~ t-\lambda_n+1 \leqslant j_k \leqslant t-\lambda_k  \textup{ for all } 1\leqslant k \leqslant i, \\
  	& ~  1 \leqslant j_s \leqslant t-\lambda_n  \textup{ for all } i+1\leqslant s \leqslant n
  	\Big\},
  	\end{align*}
  	for $ i=0,1,\cdots,n-1 $.
  	One can observe that the sets  $ \cup_{\sigma \in \tau_n} \sigma(\varOmega_i) $, with $ i=0,1,\cdots,n-1 $, are pairwise disjoint and  $ \cup_{\sigma \in \tau_n} \sigma(\varOmega_0) = \varOmega_0$. So they form a partition of $ \cup_{\sigma\in \tau_n} \sigma(B_t) $. 
  	Thus $ \cup_{\sigma\in \tau_n} \sigma(B_t) $ is equal to the union of $ \cup_{\sigma \in \tau_n} \sigma(\varOmega_i) $ for all $ i=0,1,\cdots,n-1 $. 

    Let $ \rho_i $ denote the number of lattice points in $ \cup_{\sigma \in \tau_n} \sigma(\varOmega_i) $. Then $ \rho_0=(t-n)^n $. For $ i\geqslant 1 $, we claim that
   	\begin{equation}\label{eq:rho2}
  	\rho_i=\left( \begin{array}{c} n \\i  \end{array} \right) S_i^A(n)(t-\lambda_n)^{n-i}.
  	\end{equation}
  	To see this, we separate the $ n $ coordinates $ (j_1,\cdots,j_n) $ in $ \varOmega_i $ into two distinct parts $ (j_1,\cdots,j_i) $ and $ (j_{i+1},\cdots,j_n) $ such that $ t-\lambda_n+1 \leqslant j_k \leqslant t-\lambda_k $ for all $ 1\leqslant k \leqslant i  $ and $  1 \leqslant j_s \leqslant t-\lambda_n $ for all $ i+1\leqslant s \leqslant n $. 
  	We emphasize that $ (j_1,\cdots,j_i) $ is chosen out of $ n $ positions. 
  	Define
  	\begin{align*} 
  	\varOmega_{i}^{(1)} &:= \Big\{(j_1,\cdots,j_i)\Big| ~ t-\lambda_n+1 \leqslant j_k \leqslant t-\lambda_k  \textup{ for all }  1\leqslant k \leqslant i 
  	\Big\},\\
  	 \varOmega_{i}^{(0)} & := \Big\{(j_{i+1},\cdots,j_n)\Big|  
  	 ~  1 \leqslant j_s \leqslant t-\lambda_n  \textup{ for all } i+1\leqslant s \leqslant n
  	\Big\},
  	\end{align*}
  	for each $ \varOmega_{i} $. Notice that $ \cup_{\sigma_0 \in \tau_{n-i}} \sigma_0(\varOmega_{i}^{(0)})= \varOmega_{i}^{(0)}$ and $ \varOmega_{i}^{(0)} $ has cardinality $(t-\lambda_n)^{n-i} $. On the other hand $ \varOmega_{i}^{(1)} $ is equivalent to
  	\begin{align*}\Big\{(j_1,\cdots,j_i)\Big| ~ 1 \leqslant j_k \leqslant \lambda_n-\lambda_k  \textup{ for all } 1\leqslant k \leqslant i 
  	\Big\},
  	\end{align*}
  	where $ \lambda_1=1 $. So $ \cup_{\sigma_1 \in \tau_i} \sigma_1(\varOmega_{i}^{(1)}) $ has cardinality $  S_i^A(\lambda_n) $. The choice of $(j_1,\cdots,j_i) $ now leads to the claim of \eqref{eq:rho2}. Then the desired assertion follows from the equation
  	$ S_n(t) = \sum_{i=0}^{n-1} \rho_i $, where $ \rho_i $ is given by \eqref{eq:rho2}. 
\end{proof}

With above preparations, we can determine precisely the number of pure gaps in the set $ G_0(P_1,\cdots,P_n)$ of \eqref{eq:G01} on a quotient of the Hermitian curve.
\begin{thm}\label{thm:puregap}
	Let $ q+1=mN $. For the quotient of the Hermitian curve defined by \eqref{eq:Quocurve}, if $ n> q-N $, $ G_0(P_1,\cdots,P_n)=\varnothing$, and if $ 2 \leqslant n \leqslant q-N $, then
	\begin{equation*}
	\#  G_0(P_1,\cdots,P_n)= \sum_{A=0}^{q-n-N} \left( \begin{array}{c} A+n-1 \\n-1 \end{array} \right) S_n^A(t),
	\end{equation*}
	where $  t= \ceil{\dfrac{q-A}{N}} $ and $ S_n^A(t) $ is given in Lemma~\ref{lem:Sn(t)1}.
\end{thm}
\begin{proof}
	The proof is similar to that of Theorem \ref{th:Hermitian pure gap} and so is omitted here.
\end{proof}

%
%
  

The following, viewed as an extension of Matthews \cite{matthews2001weierstrass}, is a consequence of Theorem \ref{thm:puregap} by considering the gaps and pure gaps at a pair of points.
\begin{cor}\label{cor:n=N=2}
 Let $ q=mN-1 $ and $ q-2-N\geqslant 0 $. For the quotient of the Hermitian curve defined by $ y^m=x^q+x $ over $ \mathbb{F}_{q^2} $, we have
 	\begin{align*}
 	\#  G_0(P_1,P_2)&= \frac{(q+1)(m-1)}{12}\Big((q+1)(m-1)-2m+N+7\Big)-q(m-1),\\
 	\intertext{and}
 	\#  G(P_1,P_2)&= \frac{m-1}{12}\Big((3m-1)q^2-(6m+N+5)q+3m-N-4\Big)+q(m-1).
 	\end{align*}
\end{cor}
\begin{proof}
  We now assume that $ N\geqslant 2 $ since $ N=1 $ corresponds to Hermitian curves studied by Matthews \cite{matthews2001weierstrass}. In order to calculate the cardinality of $ G_0(P_1,P_2) $, it suffices to 
   compute $ S_2^A(t) $ for $ 0\leqslant A \leqslant q-2-N $ by employing Lemma \ref{lem:Sn(t)1}.
   Since $ q-2-N=N(m-2)+(N-3) $, we consider the following cases separately.
   \begin{enumerate}
   	\item Case 1: $ A=Nl+j $ where $ 0\leqslant l \leqslant m-2 $ and $ 0\leqslant j \leqslant N-3 $. Then
   	$ t=m-l$, $\beta=N-1-j \geqslant 2$. So  
   	\begin{align*}
   	S_2^{A} (m-l ) 			 
   	=(m-l-1)^2.
   	\end{align*}
   	\item Case 2: $ A=Nl+(N-2) $ where $ 0\leqslant l \leqslant m-3 $. Then 
   	$ t=m-l$, $\beta=1 $, $ \lambda_2=2 $. So
   	\begin{align*}
   	S_2^{A} (m-l ) 			 
   	&=\sum_{i=0}^{1} \left( \begin{array}{c} 2 \\i \end{array} \right) S_i^A(2) \times (m-l-2)^{2-i}\\
   	&=(m-l-2)^{2}+2(m-l-2)\\
   	&=(m-l)(m-l-2).
   	\end{align*}		
   	\item Case 3: $ A=Nl+(N-1) $ where $ 0\leqslant l \leqslant m-3 $. Then
   	$ t=m-l-1,\beta=N \geqslant 2 $. So
   	\begin{align*}
   	S_2^A (m-l-1 ) =(m-l-2)^2.
   	\end{align*}	
    \end{enumerate} 
    It follows from Theorem \ref{thm:puregap} that 
    	\begin{align*}
    	\#  G_0(P_1,P_2)&= \sum_{l=0}^{m-2} \Bgk{\sum_{j=0}^{N-3} Nl+j+1} (m-l-1)^2+\sum_{l=0}^{m-3}(Nl+N-1)(m-l)(m-l-2)\\
    	&\phantom{=}+\sum_{l=0}^{m-2} (Nl+N)(m-l-2)^2\\
    	&=M_1+M_2+M_3,
    	\end{align*}
    where 
    	\begin{align*}
    	M_1 &= \sum_{l=0}^{m-2} \Bgk{\sum_{j=0}^{N-3} Nl+j+1} (m-l-1)^2,\\
    	M_2 &= \sum_{l=0}^{m-3}(Nl+N-1)(m-l)(m-l-2),\\
    	M_3 &= \sum_{l=0}^{m-3} (Nl+N)(m-l-2)^2.
    	\end{align*}    
    By writing $ k:=m-l-1 $, we compute directly and obtain 
    	\begin{align*}
    	M_1 &= \Bgk{N(N-2)(m-1)+\frac{(N-1)(N-2)}{2}}\sum_{k=1}^{m-1}k^2
    	       -N(N-2)\sum_{k=1}^{m-1}k^3  ,\\
    	M_2 &= q\sum_{k=1}^{m-1}k^2-N\sum_{k=1}^{m-1}k^3+N\sum_{k=1}^{m-1}k-q(m-1),\\
    	M_3 &= N(m-1)\sum_{k=1}^{m-1}k^2-N\sum_{k=1}^{m-1}k^3.
    	\end{align*}  
    It is shown by a straightforward calculation that  
    \begin{align*}
    	\#  G_0(P_1,P_2)&=M_1+M_2+M_3\\ &=\frac{(q+1)(m-1)}{12}\Big((q+1)(m-1)-2m+N+7\Big)-q(m-1),
    \end{align*}   
    concluding the first assertion of this corollary.
    
	Now we are in a position to prove the second assertion. By Lemma 6 of \cite{Yang2017Wsemi},
	\begin{align*}
	G(P_1 ) =\Big \{&
	mk+j \in \mathbb{N}  ~\Big|
	~1\leqslant j \leqslant m-1  ,   ~  0\leqslant k \leqslant q-1-Nj  \Big\},
	\end{align*}
	where $ q+1=mN $. Also, by direct computation, this yields 
	\begin{align*}
	\sum_{\alpha_1\in G(P_1) }\alpha_1 
	= \sum_{\alpha_2\in G(P_2) }\alpha_2 
	&=\sum_{j=1}^{m-1} ~\sum_{k=0}^{q-1-Nj}  (mk+j)\\
	&=\frac{ (m-1)(q-1)(2qm-m-q-1)}{12}.
	\end{align*}
	Thus it follows from Theorem \ref{thm:Gap_2point} that
	\begin{align*}
	\#  G(P_1,P_2)&=  \frac{ (m-1)(q-1)(2qm-m-q-1)}{6}-\#  G_0(P_1,P_2)\\
	&= \frac{m-1}{12}\Big((3m-1)q^2-(6m+N+5)q+3m-N-4\Big)+q(m-1),
	\end{align*}
	concluding the second assertion of this corollary.    
\end{proof}

In the following, we give some concrete examples by using our main results.
\begin{example}
	 Let $ (q,m,N)=(7,4,2) $. The quotient of the Hermitian curve becomes 
	 $ y^4=x^7+x $ of genus $ g=9 $ over $ \mathbb{F}_{49} $. 
	 It follows from Corollary \ref{cor:n=N=2} that
	 \begin{align*} 
	 \#  G_0(P_1,P_2 )= 29, \textup{ }	 
	 \#  G(P_1,P_2)=103.
	 \end{align*}
	 These results are displayed in Fig. \ref{fig:puregaps}, where the red points represent the pure gaps at $ (P_1,P_2) $ and the blue points represent the lattice points in the Weierstrass semigroup $ H(P_1,P_2) $ in the chosen area 
	 $ I:=  \{(t_1,t_2) \in \mathbb{N}_0^2 ~|~ 0\leqslant t_1,t_2 \leqslant 20  \} $. See Theorem 3 of \cite{HuYang2016Kummer}. In fact, there are $ 441 $ lattice points in $ I $ including $ 29 $ red points and $ 338 $ blue points, which indicates that the number of gaps at $ P_1,P_2 $ is $ 441-338=103 $.	These are identical with our previous values.
	  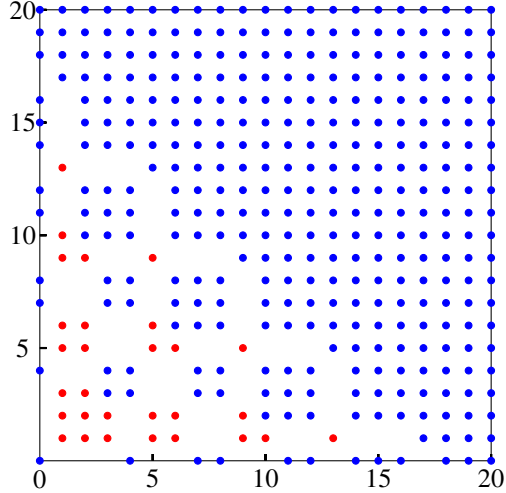
\begin{figure}[H]
	  	\centering
	  	\begin{tikzpicture}[scale=0.3]
	  	\draw [black] (0,0)--(20,0)--(20,20)--(0,20)--(0,0);
	  	
	  	\draw  node[below]{0} -- (0,0.15);
	  	\draw [thick] (5,0)--(5,.3)   node[yshift=-2ex,xshift=0ex]{5} -- (5,0);
	  	\draw [thick] (10,0)--(10,.3) node[yshift=-2ex,xshift=0ex]{10} -- (10,0);
	  	\draw [thick] (15,0)--(15,.3) node[yshift=-2ex,xshift=0ex]{15} -- (15,0);
	  	\draw [thick] (20,0)--(20,.3) node[yshift=-2ex,xshift=0ex]{20} -- (20,0);

	  	\draw [thick] (0,5)--(0.3,5)   node[yshift=0ex,xshift=-2ex]{5} -- (0,5);
	  	\draw [thick] (0,10)--(0.3,10) node[yshift=0ex,xshift=-2ex]{10} -- (0,10);
	  	\draw [thick] (0,15)--(0.3,15) node[yshift=0ex,xshift=-2ex]{15} -- (0,15);
	  	\draw [thick] (0,20)--(0.3,20) node[yshift=0ex,xshift=-2ex]{20} -- (0,20);
	  	
        \draw [red,fill] (1, 1) circle [radius=0.15]; 
        \draw [red,fill] (1, 2) circle [radius=0.15]; 
        \draw [red,fill] (1, 3) circle [radius=0.15]; 
        \draw [red,fill] (1, 5) circle [radius=0.15]; 
        \draw [red,fill] (1, 6) circle [radius=0.15]; 
        \draw [red,fill] (1, 9) circle [radius=0.15]; 
        \draw [red,fill] (1, 10) circle [radius=0.15]; 
        \draw [red,fill] (1, 13) circle [radius=0.15]; 
        \draw [red,fill] (2, 1) circle [radius=0.15]; 
        \draw [red,fill] (2, 2) circle [radius=0.15]; 
        \draw [red,fill] (2, 3) circle [radius=0.15]; 
        \draw [red,fill] (2, 5) circle [radius=0.15]; 
        \draw [red,fill] (2, 6) circle [radius=0.15]; 
        \draw [red,fill] (2, 9) circle [radius=0.15]; 
        \draw [red,fill] (3, 1) circle [radius=0.15]; 
        \draw [red,fill] (3, 2) circle [radius=0.15]; 
        \draw [red,fill] (5, 1) circle [radius=0.15]; 
        \draw [red,fill] (5, 2) circle [radius=0.15]; 
        \draw [red,fill] (5, 5) circle [radius=0.15]; 
        \draw [red,fill] (5, 6) circle [radius=0.15]; 
        \draw [red,fill] (5, 9) circle [radius=0.15]; 
        \draw [red,fill] (6, 1) circle [radius=0.15]; 
        \draw [red,fill] (6, 2) circle [radius=0.15]; 
        \draw [red,fill] (6, 5) circle [radius=0.15]; 
        \draw [red,fill] (9, 1) circle [radius=0.15]; 
        \draw [red,fill] (9, 2) circle [radius=0.15]; 
        \draw [red,fill] (9, 5) circle [radius=0.15]; 
        \draw [red,fill] (10, 1) circle [radius=0.15]; 
        \draw [red,fill] (13, 1) circle [radius=0.15]; 
 
        \draw [blue,fill] (0, 0) circle [radius=0.15]; 
        \draw [blue,fill] (0, 4) circle [radius=0.15]; 
        \draw [blue,fill] (0, 7) circle [radius=0.15]; 
        \draw [blue,fill] (0, 8) circle [radius=0.15]; 
        \draw [blue,fill] (0, 11) circle [radius=0.15]; 
        \draw [blue,fill] (0, 12) circle [radius=0.15]; 
        \draw [blue,fill] (0, 14) circle [radius=0.15]; 
        \draw [blue,fill] (0, 15) circle [radius=0.15]; 
        \draw [blue,fill] (0, 16) circle [radius=0.15]; 
        \draw [blue,fill] (0, 18) circle [radius=0.15]; 
        \draw [blue,fill] (0, 19) circle [radius=0.15]; 
        \draw [blue,fill] (0, 20) circle [radius=0.15]; 
        \draw [blue,fill] (1, 17) circle [radius=0.15]; 
        \draw [blue,fill] (1, 18) circle [radius=0.15]; 
        \draw [blue,fill] (1, 19) circle [radius=0.15]; 
        \draw [blue,fill] (1, 20) circle [radius=0.15]; 
        \draw [blue,fill] (2, 10) circle [radius=0.15]; 
        \draw [blue,fill] (2, 11) circle [radius=0.15]; 
        \draw [blue,fill] (2, 12) circle [radius=0.15]; 
        \draw [blue,fill] (2, 14) circle [radius=0.15]; 
        \draw [blue,fill] (2, 15) circle [radius=0.15]; 
        \draw [blue,fill] (2, 16) circle [radius=0.15]; 
        \draw [blue,fill] (2, 17) circle [radius=0.15]; 
        \draw [blue,fill] (2, 18) circle [radius=0.15]; 
        \draw [blue,fill] (2, 19) circle [radius=0.15]; 
        \draw [blue,fill] (2, 20) circle [radius=0.15]; 
        \draw [blue,fill] (3, 3) circle [radius=0.15]; 
        \draw [blue,fill] (3, 4) circle [radius=0.15]; 
        \draw [blue,fill] (3, 7) circle [radius=0.15]; 
        \draw [blue,fill] (3, 8) circle [radius=0.15]; 
        \draw [blue,fill] (3, 10) circle [radius=0.15]; 
        \draw [blue,fill] (3, 11) circle [radius=0.15]; 
        \draw [blue,fill] (3, 12) circle [radius=0.15]; 
        \draw [blue,fill] (3, 14) circle [radius=0.15]; 
        \draw [blue,fill] (3, 15) circle [radius=0.15]; 
        \draw [blue,fill] (3, 16) circle [radius=0.15]; 
        \draw [blue,fill] (3, 17) circle [radius=0.15]; 
        \draw [blue,fill] (3, 18) circle [radius=0.15]; 
        \draw [blue,fill] (3, 19) circle [radius=0.15]; 
        \draw [blue,fill] (3, 20) circle [radius=0.15]; 
        \draw [blue,fill] (4, 0) circle [radius=0.15]; 
        \draw [blue,fill] (4, 3) circle [radius=0.15]; 
        \draw [blue,fill] (4, 4) circle [radius=0.15]; 
        \draw [blue,fill] (4, 7) circle [radius=0.15]; 
        \draw [blue,fill] (4, 8) circle [radius=0.15]; 
        \draw [blue,fill] (4, 10) circle [radius=0.15]; 
        \draw [blue,fill] (4, 11) circle [radius=0.15]; 
        \draw [blue,fill] (4, 12) circle [radius=0.15]; 
        \draw [blue,fill] (4, 14) circle [radius=0.15]; 
        \draw [blue,fill] (4, 15) circle [radius=0.15]; 
        \draw [blue,fill] (4, 16) circle [radius=0.15]; 
        \draw [blue,fill] (4, 17) circle [radius=0.15]; 
        \draw [blue,fill] (4, 18) circle [radius=0.15]; 
        \draw [blue,fill] (4, 19) circle [radius=0.15]; 
        \draw [blue,fill] (4, 20) circle [radius=0.15]; 
        \draw [blue,fill] (5, 13) circle [radius=0.15]; 
        \draw [blue,fill] (5, 14) circle [radius=0.15]; 
        \draw [blue,fill] (5, 15) circle [radius=0.15]; 
        \draw [blue,fill] (5, 16) circle [radius=0.15]; 
        \draw [blue,fill] (5, 17) circle [radius=0.15]; 
        \draw [blue,fill] (5, 18) circle [radius=0.15]; 
        \draw [blue,fill] (5, 19) circle [radius=0.15]; 
        \draw [blue,fill] (5, 20) circle [radius=0.15]; 
        \draw [blue,fill] (6, 6) circle [radius=0.15]; 
        \draw [blue,fill] (6, 7) circle [radius=0.15]; 
        \draw [blue,fill] (6, 8) circle [radius=0.15]; 
        \draw [blue,fill] (6, 10) circle [radius=0.15]; 
        \draw [blue,fill] (6, 11) circle [radius=0.15]; 
        \draw [blue,fill] (6, 12) circle [radius=0.15]; 
        \draw [blue,fill] (6, 13) circle [radius=0.15]; 
        \draw [blue,fill] (6, 14) circle [radius=0.15]; 
        \draw [blue,fill] (6, 15) circle [radius=0.15]; 
        \draw [blue,fill] (6, 16) circle [radius=0.15]; 
        \draw [blue,fill] (6, 17) circle [radius=0.15]; 
        \draw [blue,fill] (6, 18) circle [radius=0.15]; 
        \draw [blue,fill] (6, 19) circle [radius=0.15]; 
        \draw [blue,fill] (6, 20) circle [radius=0.15]; 
        \draw [blue,fill] (7, 0) circle [radius=0.15]; 
        \draw [blue,fill] (7, 3) circle [radius=0.15]; 
        \draw [blue,fill] (7, 4) circle [radius=0.15]; 
        \draw [blue,fill] (7, 6) circle [radius=0.15]; 
        \draw [blue,fill] (7, 7) circle [radius=0.15]; 
        \draw [blue,fill] (7, 8) circle [radius=0.15]; 
        \draw [blue,fill] (7, 10) circle [radius=0.15]; 
        \draw [blue,fill] (7, 11) circle [radius=0.15]; 
        \draw [blue,fill] (7, 12) circle [radius=0.15]; 
        \draw [blue,fill] (7, 13) circle [radius=0.15]; 
        \draw [blue,fill] (7, 14) circle [radius=0.15]; 
        \draw [blue,fill] (7, 15) circle [radius=0.15]; 
        \draw [blue,fill] (7, 16) circle [radius=0.15]; 
        \draw [blue,fill] (7, 17) circle [radius=0.15]; 
        \draw [blue,fill] (7, 18) circle [radius=0.15]; 
        \draw [blue,fill] (7, 19) circle [radius=0.15]; 
        \draw [blue,fill] (7, 20) circle [radius=0.15]; 
        \draw [blue,fill] (8, 0) circle [radius=0.15]; 
        \draw [blue,fill] (8, 3) circle [radius=0.15]; 
        \draw [blue,fill] (8, 4) circle [radius=0.15]; 
        \draw [blue,fill] (8, 6) circle [radius=0.15]; 
        \draw [blue,fill] (8, 7) circle [radius=0.15]; 
        \draw [blue,fill] (8, 8) circle [radius=0.15]; 
        \draw [blue,fill] (8, 10) circle [radius=0.15]; 
        \draw [blue,fill] (8, 11) circle [radius=0.15]; 
        \draw [blue,fill] (8, 12) circle [radius=0.15]; 
        \draw [blue,fill] (8, 13) circle [radius=0.15]; 
        \draw [blue,fill] (8, 14) circle [radius=0.15]; 
        \draw [blue,fill] (8, 15) circle [radius=0.15]; 
        \draw [blue,fill] (8, 16) circle [radius=0.15]; 
        \draw [blue,fill] (8, 17) circle [radius=0.15]; 
        \draw [blue,fill] (8, 18) circle [radius=0.15]; 
        \draw [blue,fill] (8, 19) circle [radius=0.15]; 
        \draw [blue,fill] (8, 20) circle [radius=0.15]; 
        \draw [blue,fill] (9, 9) circle [radius=0.15]; 
        \draw [blue,fill] (9, 10) circle [radius=0.15]; 
        \draw [blue,fill] (9, 11) circle [radius=0.15]; 
        \draw [blue,fill] (9, 12) circle [radius=0.15]; 
        \draw [blue,fill] (9, 13) circle [radius=0.15]; 
        \draw [blue,fill] (9, 14) circle [radius=0.15]; 
        \draw [blue,fill] (9, 15) circle [radius=0.15]; 
        \draw [blue,fill] (9, 16) circle [radius=0.15]; 
        \draw [blue,fill] (9, 17) circle [radius=0.15]; 
        \draw [blue,fill] (9, 18) circle [radius=0.15]; 
        \draw [blue,fill] (9, 19) circle [radius=0.15]; 
        \draw [blue,fill] (9, 20) circle [radius=0.15]; 
        \draw [blue,fill] (10, 2) circle [radius=0.15]; 
        \draw [blue,fill] (10, 3) circle [radius=0.15]; 
        \draw [blue,fill] (10, 4) circle [radius=0.15]; 
        \draw [blue,fill] (10, 6) circle [radius=0.15]; 
        \draw [blue,fill] (10, 7) circle [radius=0.15]; 
        \draw [blue,fill] (10, 8) circle [radius=0.15]; 
        \draw [blue,fill] (10, 9) circle [radius=0.15]; 
        \draw [blue,fill] (10, 10) circle [radius=0.15]; 
        \draw [blue,fill] (10, 11) circle [radius=0.15]; 
        \draw [blue,fill] (10, 12) circle [radius=0.15]; 
        \draw [blue,fill] (10, 13) circle [radius=0.15]; 
        \draw [blue,fill] (10, 14) circle [radius=0.15]; 
        \draw [blue,fill] (10, 15) circle [radius=0.15]; 
        \draw [blue,fill] (10, 16) circle [radius=0.15]; 
        \draw [blue,fill] (10, 17) circle [radius=0.15]; 
        \draw [blue,fill] (10, 18) circle [radius=0.15]; 
        \draw [blue,fill] (10, 19) circle [radius=0.15]; 
        \draw [blue,fill] (10, 20) circle [radius=0.15]; 
        \draw [blue,fill] (11, 0) circle [radius=0.15]; 
        \draw [blue,fill] (11, 2) circle [radius=0.15]; 
        \draw [blue,fill] (11, 3) circle [radius=0.15]; 
        \draw [blue,fill] (11, 4) circle [radius=0.15]; 
        \draw [blue,fill] (11, 6) circle [radius=0.15]; 
        \draw [blue,fill] (11, 7) circle [radius=0.15]; 
        \draw [blue,fill] (11, 8) circle [radius=0.15]; 
        \draw [blue,fill] (11, 9) circle [radius=0.15]; 
        \draw [blue,fill] (11, 10) circle [radius=0.15]; 
        \draw [blue,fill] (11, 11) circle [radius=0.15]; 
        \draw [blue,fill] (11, 12) circle [radius=0.15]; 
        \draw [blue,fill] (11, 13) circle [radius=0.15]; 
        \draw [blue,fill] (11, 14) circle [radius=0.15]; 
        \draw [blue,fill] (11, 15) circle [radius=0.15]; 
        \draw [blue,fill] (11, 16) circle [radius=0.15]; 
        \draw [blue,fill] (11, 17) circle [radius=0.15]; 
        \draw [blue,fill] (11, 18) circle [radius=0.15]; 
        \draw [blue,fill] (11, 19) circle [radius=0.15]; 
        \draw [blue,fill] (11, 20) circle [radius=0.15]; 
        \draw [blue,fill] (12, 0) circle [radius=0.15]; 
        \draw [blue,fill] (12, 2) circle [radius=0.15]; 
        \draw [blue,fill] (12, 3) circle [radius=0.15]; 
        \draw [blue,fill] (12, 4) circle [radius=0.15]; 
        \draw [blue,fill] (12, 6) circle [radius=0.15]; 
        \draw [blue,fill] (12, 7) circle [radius=0.15]; 
        \draw [blue,fill] (12, 8) circle [radius=0.15]; 
        \draw [blue,fill] (12, 9) circle [radius=0.15]; 
        \draw [blue,fill] (12, 10) circle [radius=0.15]; 
        \draw [blue,fill] (12, 11) circle [radius=0.15]; 
        \draw [blue,fill] (12, 12) circle [radius=0.15]; 
        \draw [blue,fill] (12, 13) circle [radius=0.15]; 
        \draw [blue,fill] (12, 14) circle [radius=0.15]; 
        \draw [blue,fill] (12, 15) circle [radius=0.15]; 
        \draw [blue,fill] (12, 16) circle [radius=0.15]; 
        \draw [blue,fill] (12, 17) circle [radius=0.15]; 
        \draw [blue,fill] (12, 18) circle [radius=0.15]; 
        \draw [blue,fill] (12, 19) circle [radius=0.15]; 
        \draw [blue,fill] (12, 20) circle [radius=0.15]; 
        \draw [blue,fill] (13, 5) circle [radius=0.15]; 
        \draw [blue,fill] (13, 6) circle [radius=0.15]; 
        \draw [blue,fill] (13, 7) circle [radius=0.15]; 
        \draw [blue,fill] (13, 8) circle [radius=0.15]; 
        \draw [blue,fill] (13, 9) circle [radius=0.15]; 
        \draw [blue,fill] (13, 10) circle [radius=0.15]; 
        \draw [blue,fill] (13, 11) circle [radius=0.15]; 
        \draw [blue,fill] (13, 12) circle [radius=0.15]; 
        \draw [blue,fill] (13, 13) circle [radius=0.15]; 
        \draw [blue,fill] (13, 14) circle [radius=0.15]; 
        \draw [blue,fill] (13, 15) circle [radius=0.15]; 
        \draw [blue,fill] (13, 16) circle [radius=0.15]; 
        \draw [blue,fill] (13, 17) circle [radius=0.15]; 
        \draw [blue,fill] (13, 18) circle [radius=0.15]; 
        \draw [blue,fill] (13, 19) circle [radius=0.15]; 
        \draw [blue,fill] (13, 20) circle [radius=0.15]; 
        \draw [blue,fill] (14, 0) circle [radius=0.15]; 
        \draw [blue,fill] (14, 2) circle [radius=0.15]; 
        \draw [blue,fill] (14, 3) circle [radius=0.15]; 
        \draw [blue,fill] (14, 4) circle [radius=0.15]; 
        \draw [blue,fill] (14, 5) circle [radius=0.15]; 
        \draw [blue,fill] (14, 6) circle [radius=0.15]; 
        \draw [blue,fill] (14, 7) circle [radius=0.15]; 
        \draw [blue,fill] (14, 8) circle [radius=0.15]; 
        \draw [blue,fill] (14, 9) circle [radius=0.15]; 
        \draw [blue,fill] (14, 10) circle [radius=0.15]; 
        \draw [blue,fill] (14, 11) circle [radius=0.15]; 
        \draw [blue,fill] (14, 12) circle [radius=0.15]; 
        \draw [blue,fill] (14, 13) circle [radius=0.15]; 
        \draw [blue,fill] (14, 14) circle [radius=0.15]; 
        \draw [blue,fill] (14, 15) circle [radius=0.15]; 
        \draw [blue,fill] (14, 16) circle [radius=0.15]; 
        \draw [blue,fill] (14, 17) circle [radius=0.15]; 
        \draw [blue,fill] (14, 18) circle [radius=0.15]; 
        \draw [blue,fill] (14, 19) circle [radius=0.15]; 
        \draw [blue,fill] (14, 20) circle [radius=0.15]; 
        \draw [blue,fill] (15, 0) circle [radius=0.15]; 
        \draw [blue,fill] (15, 2) circle [radius=0.15]; 
        \draw [blue,fill] (15, 3) circle [radius=0.15]; 
        \draw [blue,fill] (15, 4) circle [radius=0.15]; 
        \draw [blue,fill] (15, 5) circle [radius=0.15]; 
        \draw [blue,fill] (15, 6) circle [radius=0.15]; 
        \draw [blue,fill] (15, 7) circle [radius=0.15]; 
        \draw [blue,fill] (15, 8) circle [radius=0.15]; 
        \draw [blue,fill] (15, 9) circle [radius=0.15]; 
        \draw [blue,fill] (15, 10) circle [radius=0.15]; 
        \draw [blue,fill] (15, 11) circle [radius=0.15]; 
        \draw [blue,fill] (15, 12) circle [radius=0.15]; 
        \draw [blue,fill] (15, 13) circle [radius=0.15]; 
        \draw [blue,fill] (15, 14) circle [radius=0.15]; 
        \draw [blue,fill] (15, 15) circle [radius=0.15]; 
        \draw [blue,fill] (15, 16) circle [radius=0.15]; 
        \draw [blue,fill] (15, 17) circle [radius=0.15]; 
        \draw [blue,fill] (15, 18) circle [radius=0.15]; 
        \draw [blue,fill] (15, 19) circle [radius=0.15]; 
        \draw [blue,fill] (15, 20) circle [radius=0.15]; 
        \draw [blue,fill] (16, 0) circle [radius=0.15]; 
        \draw [blue,fill] (16, 2) circle [radius=0.15]; 
        \draw [blue,fill] (16, 3) circle [radius=0.15]; 
        \draw [blue,fill] (16, 4) circle [radius=0.15]; 
        \draw [blue,fill] (16, 5) circle [radius=0.15]; 
        \draw [blue,fill] (16, 6) circle [radius=0.15]; 
        \draw [blue,fill] (16, 7) circle [radius=0.15]; 
        \draw [blue,fill] (16, 8) circle [radius=0.15]; 
        \draw [blue,fill] (16, 9) circle [radius=0.15]; 
        \draw [blue,fill] (16, 10) circle [radius=0.15]; 
        \draw [blue,fill] (16, 11) circle [radius=0.15]; 
        \draw [blue,fill] (16, 12) circle [radius=0.15]; 
        \draw [blue,fill] (16, 13) circle [radius=0.15]; 
        \draw [blue,fill] (16, 14) circle [radius=0.15]; 
        \draw [blue,fill] (16, 15) circle [radius=0.15]; 
        \draw [blue,fill] (16, 16) circle [radius=0.15]; 
        \draw [blue,fill] (16, 17) circle [radius=0.15]; 
        \draw [blue,fill] (16, 18) circle [radius=0.15]; 
        \draw [blue,fill] (16, 19) circle [radius=0.15]; 
        \draw [blue,fill] (16, 20) circle [radius=0.15]; 
        \draw [blue,fill] (17, 1) circle [radius=0.15]; 
        \draw [blue,fill] (17, 2) circle [radius=0.15]; 
        \draw [blue,fill] (17, 3) circle [radius=0.15]; 
        \draw [blue,fill] (17, 4) circle [radius=0.15]; 
        \draw [blue,fill] (17, 5) circle [radius=0.15]; 
        \draw [blue,fill] (17, 6) circle [radius=0.15]; 
        \draw [blue,fill] (17, 7) circle [radius=0.15]; 
        \draw [blue,fill] (17, 8) circle [radius=0.15]; 
        \draw [blue,fill] (17, 9) circle [radius=0.15]; 
        \draw [blue,fill] (17, 10) circle [radius=0.15]; 
        \draw [blue,fill] (17, 11) circle [radius=0.15]; 
        \draw [blue,fill] (17, 12) circle [radius=0.15]; 
        \draw [blue,fill] (17, 13) circle [radius=0.15]; 
        \draw [blue,fill] (17, 14) circle [radius=0.15]; 
        \draw [blue,fill] (17, 15) circle [radius=0.15]; 
        \draw [blue,fill] (17, 16) circle [radius=0.15]; 
        \draw [blue,fill] (17, 17) circle [radius=0.15]; 
        \draw [blue,fill] (17, 18) circle [radius=0.15]; 
        \draw [blue,fill] (17, 19) circle [radius=0.15]; 
        \draw [blue,fill] (17, 20) circle [radius=0.15]; 
        \draw [blue,fill] (18, 0) circle [radius=0.15]; 
        \draw [blue,fill] (18, 1) circle [radius=0.15]; 
        \draw [blue,fill] (18, 2) circle [radius=0.15]; 
        \draw [blue,fill] (18, 3) circle [radius=0.15]; 
        \draw [blue,fill] (18, 4) circle [radius=0.15]; 
        \draw [blue,fill] (18, 5) circle [radius=0.15]; 
        \draw [blue,fill] (18, 6) circle [radius=0.15]; 
        \draw [blue,fill] (18, 7) circle [radius=0.15]; 
        \draw [blue,fill] (18, 8) circle [radius=0.15]; 
        \draw [blue,fill] (18, 9) circle [radius=0.15]; 
        \draw [blue,fill] (18, 10) circle [radius=0.15]; 
        \draw [blue,fill] (18, 11) circle [radius=0.15]; 
        \draw [blue,fill] (18, 12) circle [radius=0.15]; 
        \draw [blue,fill] (18, 13) circle [radius=0.15]; 
        \draw [blue,fill] (18, 14) circle [radius=0.15]; 
        \draw [blue,fill] (18, 15) circle [radius=0.15]; 
        \draw [blue,fill] (18, 16) circle [radius=0.15]; 
        \draw [blue,fill] (18, 17) circle [radius=0.15]; 
        \draw [blue,fill] (18, 18) circle [radius=0.15]; 
        \draw [blue,fill] (18, 19) circle [radius=0.15]; 
        \draw [blue,fill] (18, 20) circle [radius=0.15]; 
        \draw [blue,fill] (19, 0) circle [radius=0.15]; 
        \draw [blue,fill] (19, 1) circle [radius=0.15]; 
        \draw [blue,fill] (19, 2) circle [radius=0.15]; 
        \draw [blue,fill] (19, 3) circle [radius=0.15]; 
        \draw [blue,fill] (19, 4) circle [radius=0.15]; 
        \draw [blue,fill] (19, 5) circle [radius=0.15]; 
        \draw [blue,fill] (19, 6) circle [radius=0.15]; 
        \draw [blue,fill] (19, 7) circle [radius=0.15]; 
        \draw [blue,fill] (19, 8) circle [radius=0.15]; 
        \draw [blue,fill] (19, 9) circle [radius=0.15]; 
        \draw [blue,fill] (19, 10) circle [radius=0.15]; 
        \draw [blue,fill] (19, 11) circle [radius=0.15]; 
        \draw [blue,fill] (19, 12) circle [radius=0.15]; 
        \draw [blue,fill] (19, 13) circle [radius=0.15]; 
        \draw [blue,fill] (19, 14) circle [radius=0.15]; 
        \draw [blue,fill] (19, 15) circle [radius=0.15]; 
        \draw [blue,fill] (19, 16) circle [radius=0.15]; 
        \draw [blue,fill] (19, 17) circle [radius=0.15]; 
        \draw [blue,fill] (19, 18) circle [radius=0.15]; 
        \draw [blue,fill] (19, 19) circle [radius=0.15]; 
        \draw [blue,fill] (19, 20) circle [radius=0.15]; 
        \draw [blue,fill] (20, 0) circle [radius=0.15]; 
        \draw [blue,fill] (20, 1) circle [radius=0.15]; 
        \draw [blue,fill] (20, 2) circle [radius=0.15]; 
        \draw [blue,fill] (20, 3) circle [radius=0.15]; 
        \draw [blue,fill] (20, 4) circle [radius=0.15]; 
        \draw [blue,fill] (20, 5) circle [radius=0.15]; 
        \draw [blue,fill] (20, 6) circle [radius=0.15]; 
        \draw [blue,fill] (20, 7) circle [radius=0.15]; 
        \draw [blue,fill] (20, 8) circle [radius=0.15]; 
        \draw [blue,fill] (20, 9) circle [radius=0.15]; 
        \draw [blue,fill] (20, 10) circle [radius=0.15]; 
        \draw [blue,fill] (20, 11) circle [radius=0.15]; 
        \draw [blue,fill] (20, 12) circle [radius=0.15]; 
        \draw [blue,fill] (20, 13) circle [radius=0.15]; 
        \draw [blue,fill] (20, 14) circle [radius=0.15]; 
        \draw [blue,fill] (20, 15) circle [radius=0.15]; 
        \draw [blue,fill] (20, 16) circle [radius=0.15]; 
        \draw [blue,fill] (20, 17) circle [radius=0.15]; 
        \draw [blue,fill] (20, 18) circle [radius=0.15]; 
        \draw [blue,fill] (20, 19) circle [radius=0.15]; 
        \draw [blue,fill] (20, 20) circle [radius=0.15];
 	
	  	\end{tikzpicture}
	  	
	  	\protect\caption{The pure gaps at $ (P_1,P_2) $ }
	  	\label{fig:puregaps}
	  \end{figure}
\end{example}
\begin{example}
 Let $ (q,m,N)=(8,3,3) $.  Let us consider the quotient of the Hermitian curve  
$ y^3=x^8+x $  over $ \mathbb{F}_{64} $. 
It follows from Theorem \ref{thm:puregap} that
	\begin{equation}\label{eq:exam2}
	\#  G_0(P_1,P_2,P_3)= \sum_{A=0}^{2} \left( \begin{array}{c} A+2 \\2 \end{array} \right) S_3^A (t ),
	\end{equation}
	 where $  t=\ceil{\frac{8-A}{3}} $. So it suffices to compute $ S_3^A (t ) $
	 for $ A=0,1,2 $ as showed in details here.
\begin{enumerate}
	\item Case 1: $ A=0 $. We compute from Lemma \ref{lem:Sn(t)1} that
	$ t=3,\beta=2 $ and $ \lambda_3=2 $. Since $ t> \lambda_3$, we have
	\begin{align*}
	 S_3^0 (3 ) 			 
	=\sum_{i=0}^{2} \left( \begin{array}{c} 3 \\i \end{array} \right) S_i^0(2)
	=1+3\times 1+3\times 1=7.
	\end{align*}
	\item Case 2: $ A=1 $. By Lemma \ref{lem:Sn(t)1}, 
	$ t=3,\beta=1 $ and $ \lambda_3=2 $. So
		\begin{align*}
		S_3^1 (3 ) 			 
		=\sum_{i=0}^{2} \left( \begin{array}{c} 3 \\i \end{array} \right) S_i^1(2)
		=1+3\times 1+3\times 0=4.
		\end{align*}
	\item Case 3: $ A=2 $. Again from Lemma \ref{lem:Sn(t)1}, we have
	$ t=2,\beta=3 $ and $ \lambda_3=1 $. So
	\begin{align*}
	S_3^2 (2 ) 			 
	=\sum_{i=0}^{2} \left( \begin{array}{c} 3 \\i \end{array} \right) S_i^2(1)
	=1+3\times 0+3\times 0=1.
	\end{align*}
	\end{enumerate}
	Therefore, the number of pure gaps at $  (P_1,P_2,P_3) $ is
	$ 	\#  G_0(P_1,P_2,P_3) =25 $ by \eqref{eq:exam2}.

On the other hand, by Theorem \ref{thm:Weier},
we have  	
 	\begin{align*}
 	G_0(P_1,P_2,P_3)=\left\{\begin{array}{ll}
  	 &(1, 1, 1), 
  	 (1, 1, 2), 
  	 (1, 1, 4), 
  	 (1, 1, 5), 
  	 (1, 1, 7), \\
  	 & 
  	 (1, 2, 1), 
  	 (1, 2, 2),
  	 (1, 2, 4), 
  	 (1, 4, 1), 
  	 (1, 4, 2),\\ 
  	 &(1, 4, 4), 
  	 (1, 5, 1), 
  	 (1, 7, 1), 
  	 (2, 1, 1), 
  	 (2, 1, 2), \\
  	 &(2, 1, 4), 
  	 (2, 2, 1), 
  	 (2, 4, 1), 
  	 (4, 1, 1), 
  	 (4, 1, 2), \\
  	 &(4, 1, 4), 
  	 (4, 2, 1), 
  	 (4, 4, 1), 
  	 (5, 1, 1), 
  	 (7, 1, 1)
 	\end{array}
 	\right\} .
 \end{align*}
 So the cardinality is $ \# G_0(P_1,P_2,P_3) =25 $, which coincides with our calculation.
\end{example}

\ifCLASSOPTIONcaptionsoff
  \newpage
\fi

%

\end{document}